\documentclass[conference]{IEEEtran}
\IEEEoverridecommandlockouts

\usepackage{fancyhdr}
\usepackage{booktabs} 
\usepackage{amssymb}
\usepackage{amsmath,amsthm}
\usepackage{graphicx}
\usepackage[tight,footnotesize]{subfigure}
\usepackage{multirow}
\usepackage{listings}
\usepackage{color}
\usepackage{caption}
\usepackage{algorithm}
\usepackage{algorithmicx}
\usepackage[noend]{algpseudocode}
\usepackage{xspace}
\usepackage{url}
\usepackage{breakurl}

\usepackage{flushend}

\newcommand{\para}[1]{\xspace \smallskip \noindent\textbf{#1} \ }

\makeatletter
\newcounter{protocol}

\makeatother

\lstdefinestyle{mystyle}{
	keywordstyle=\color{magenta},
	basicstyle=\ttfamily\footnotesize,
	breakatwhitespace=false,         
	breaklines=true,                 
	captionpos=b,                    
	keepspaces=true,                 
	numbers=left,                    
	numbersep=5pt,                  
	showspaces=false,                
	showstringspaces=false,
	showtabs=flase,                  
	tabsize=2,
}

\newtheorem{theorem}{Theorem}

\newtheorem{definition}{Definition}

\begin{document}
\title{On Finding Dense Subgraphs in Bipartite Graphs: Linear Algorithms with Applications to Fraud Detection}

\author{\IEEEauthorblockN{ Yikun Ban}
\IEEEauthorblockA{Peking University}
\and

\IEEEauthorblockN{ Yitao Duan}
\IEEEauthorblockA{Netease Youdao Inc.}

}

\maketitle
\vspace{-0.5em}
\begin{abstract}
Detecting dense subgraphs from large graphs is a core component in many applications, ranging from social networks mining, bioinformatics, to online fraud detection. In this paper, we focus on mining dense subgraphs in a bipartite graph. The work is motivated by the task of fraud detection that can often be formulated as mining a bipartite graph formed by the source nodes (followers, customers) and target nodes (followees, products, etc.) for malicious patterns. We introduce a new restricted biclique problem,  Maximal Half Isolated Biclique (MHI Biclique), and show that the problem finds immediate applications in fraud detection. We prove that, unlike many other biclique problems such as the maximum edge biclique problem that are known to be NP-Complete, the MHI Biclique problem admits a \emph{linear} time solution. We provide a novel algorithm S-tree, and its extension, S-forest, that solves the problem efficiently. We also demonstrate that the algorithms are robust against deliberate camouflaging and other perturbations. Furthermore, our approach can automatically combine and prioritize multiple features, reducing the need for feature engineering while maintaining security against unseen attacks. Extensive experiments on several public and proprietary datasets demonstrate that S-tree/S-forest outperforms strong rivals across all configurations, becoming the new state of the art in fraud detection.

\end{abstract}


\vspace{-0.5em}
\section{Introduction}
\vspace{-0.5em}

Dense subgraph detection is a major topic in algorithmic graph theory and their applications can be found in a wide range of real-world scenarios such as social network analysis \cite{Flake2000Efficient}, link spam detection  \cite{MZOOM}, and bioinformatics \cite{GENE}, etc. In many graphs that model interactions among users (e.g., a social network) or between users and a platform (e.g., an e-commerce site), dense subgraphs tend to signal interesting phenomena or indicate a group of accomplices.   

In this paper, we focus on mining dense subgraphs in a bipartite graph. The work is motivated by the task of fraud detection that can often be formulated as mining a bipartite graph formed by source nodes (followers, customers) and target nodes (followees, products, etc.) for malicious patterns. Fraud is a serious problem in our society, especially in online social applications (OSA) such as Amazon, Twitter, Weibo, etc. The interactions between users and the platform can be modeled as a bipartite graph $G$ between source nodes (accounts) and target nodes (e.g., followee, products, etc.). In a common scenario, fraudsters try to manipulate the ranking of some specific objects on OSA by creating fake edges (follows, reviews, etc.) from the accounts they control. Since fraudsters are bounded in resources, and they are trying to maximize financial gains, fraud groups (fake accounts and their collusive customers) inevitably exhibit synchronized behavior, forming dense subgraphs in $G$.

Fraud detection is an area of active research. Many works try to detect such dense subgraphs using belief propagation(BP)\cite{FRAUDEAGLE, NETPROBE}, HITS\cite{HITS}-like ideas \cite{COMBATING, CATCHSYNC}, singular value decomposition (SVD)\cite{SPOKEN,FBOX}, or greedy pruning \cite{FRAUDAR,CROSSSPOT,MZOOM,DCUBE}. In its simplest forms, fraud groups are likely to form bicliques (a complete bipartite graph) in $G$. Unfortunately, finding all maximal bicliques cannot be solved in polynomial time. This is further complicated by fraudsters trying to disguise themselves as normal users, a tactic known as ``camouflaging''. This is typically done by adding edges from fraudulent source nodes to legit target nodes. Existing approaches are either inefficient or ineffective in the presence of camouflage.

Another challenge in fraud detection is that fraud groups display synchronized behavior only on certain dimensions (e.g., IP address, phones, etc.), which are unknown in advance and may change across groups and time. Treating all features equally will be misled by irrelevant features. Effective fraud detection often requires heavy feature engineering.

\para{Our Contributions} We formulate fraud detection as a restricted case of biclique mining and give it a rigorous treatment. We observe that most fraud subgraphs are \emph{half isolated}, meaning that they connect to other nodes only through at most one side. For example, it is easy for fraudsters to create edges from the accounts they control to other nodes (including legit nodes), but is rare for fraud nodes to gain edges from legit users. We introduce the 
Maximal Half Isolated (MHI) Biclique problem which admits a linear time solution. By utilizing an efficient and scalable algorithm, the model can not only run faster, handle more data, but also detect all MHI biclques and frauds more accurately. This is in contrast to other solutions relying on heuristics or approximation. Concretely, our contributions are summarized as follows.    

\begin{itemize}
\item \textbf{[A New Graph Problem]}. We introduce a new restricted biclique problem, Maximal Half Isolated Biclique (MHI Biclique), and show that the problem finds immediate applications in fraud detection. 

\item \textbf{[A Linear Algorithm]}. We propose a novel data structure, S-tree, and a mining algorithm, that solves the MHI Biclique problem in linear time. The algorithm is effective for detecting unbalanced dense sub-bipartite graphs. We provide theoretical proofs regarding the algorithm's effectiveness and robustness against adversarial perturbations. 

\item \textbf{[Practical Algorithms for Fraud Detection]}. Based on the S-tree-based MHI Biclique problem solver, we introduce a new algorithm that detects near bicliques and can be used to catch a wide range of fraud groups. The algorithm has several advantages compared with other fraud detection algorithms, summarized in Table \ref{tab.compare}.

\item \textbf{[Automatic Feature Prioritization]}. We further extend S-tree to S-forest to handle multimodal data. S-forest can utilize \emph{both} structural or graph data representing the relations between objects and attribute data characterizing individual objects. Furthermore, S-forest can automatically combine and prioritize multiple features, reducing the need for feature engineering.

\item \textbf{[Effectiveness on Real-world Data]}. We conducted extensive experiments on thirteen real-world datasets, including twelve public datasets and one proprietary data collected from the production system log of a major e-commerce vendor. Our solution outperforms strong rivals across all configurations, becoming the new state of the art in fraud detection.
\end{itemize}

%
%
%


\begin{table} 
\footnotesize 
\caption{S-tree/S-forest v.s. other fraud detection methods}\label{tab.compare}
\vspace{-1em}
\centering
\begin{tabular}{c|cccccc|c}
\toprule
  & \rotatebox{90}{Fraudar\cite{FRAUDAR}} &\rotatebox{90}{Spoken \cite{SPOKEN} }&\rotatebox{90}{CatchSync \cite{CATCHSYNC}} &\rotatebox{90}{CrossSpot \cite{CROSSSPOT}}&  \rotatebox{90}{Fbox\cite{FBOX}}&\rotatebox{90}{M-zoom\cite{MZOOM}} &\rotatebox{90}{S-tree/S-forest} \\
\midrule
 Handling multimodel data ?&$\times$&$\times$&$\times$&$\surd $&$\times$&$\surd$ &$\surd$\\
Camouflage resistant?& $\surd$ & $\times$  & $\times$ & ? & $\times$ & ? & $\surd$ \\
Linear time?& $\times$ & $\times$  & $\times$ & $\times$ & $\times$ &$\times$& $\surd$ \\
\bottomrule
\end{tabular}
\centering
\vspace{-2.5em}
\end{table}

\vspace{-0.5em}
\section{Problem Definition}\label{sec:prob}
\vspace{-0.5em}
\subsection{Graph modeling}

Consider a bipartite graph $G = (V, E)$ with vertex set $V$ and edge set $E$ where $V= N \cup M$  consists of two disjoint subsets: a set of \emph{source nodes}  $N = \{n_1,...,n_{|N|}\}$  and a set of \emph{target nodes} $M = \{m_1, ...,m_{|M|}\}$. $(n, m) \in E$ denotes an edge between $n \in N$ and $m \in M$. We formulate the fraud detection problem in two modes. In both cases, $N$ are user accounts, some of which are controlled by fraudsters. The two differ by how they construct $M$. In the first case, elements in $M$ are \emph{objects}, the entities whose ranking the fraudsters try to manipulate. Examples include followees on Twitter, products on Amazon, etc. In the second mode, $M$ consists of \emph{resources} that the users use to interact with the platform (e.g., IP addresses, timestamps, device IDs, etc). This mode typically involves multiple dimensions and we model them using multiple bipartite graphs. We call the first mode account-object bipartite graph (AOBG) and the second account-resource  graph (ARBG).
As will be shown later, fraud activities display remarkable synchrony in both cases and our scheme handles them in a uniform way. For any specific fraud detection problem, practitioners are free to choose either mode or combine both to achieve optimal results.

Our goal is to detect a subset of source nodes $\mathcal{N}$ that are likely to be involved in fraud activities. Table \ref{tab1} gives the list of symbols we use throughout the paper.

\subsection{Fraud Attack}\label{Synchrony}

\begin{table}[t]
\caption{Symbols and Definitions}\label{tab1}
\vspace{-1em}
 \centering
\begin{tabular}{l|l}
\toprule
Symbols &  Definition\\
\midrule
$N$ &  The set of source nodes, $N = \{n_1,...,n_{|N|}\}$\\ 
$M$ &  The set of target nodes, $M = \{m_1, ...,m_{|M|}\}$\\
$G$ & The bipartite graph, $G =(V, E)$, $V = M \cup N$ \\
$\mathcal{N}$ & A subset of $N$\\
$\mathcal{M}$ & A subset of $M$\\
$\mathcal{G}$ & A subgraph of $G$ \\
$B(m)$&  The basket of $m$, $B(m) = \{m, I(m), f(m)\}$ \\
$T$ & The suspiciousness tree (S-tree)\\
$x$ & A node of $T$, which has three fields: \emph{sn}, \emph{sus}, and \emph{tn}\\
\bottomrule
\end{tabular}
\vspace{-2em}
\end{table}

Most frauds are conducted for financial gains. To reduce cost, fraudsters operate on limited resources (e.g., phone numbers, IP addresses, time etc.) and also resort to economies of scale to maximize profits \cite{SYNCHROTRAP}. As a result, fraudulent accounts exhibit unusually synchronized patterns compared to the legit users, forming dense subgraphs on $G$. This has been the key signal that many fraud detection works\cite{CATCHSYNC, COPYCATCH, FRAUDAR, CROSSSPOT} try to capture. Similarly, with multi-dimensional data, fraud activities often form dense regions in tensor models, which has been observed in network intrusion\cite{MZOOM}, bot activities\cite{MZOOM}, and genetic applications\cite{GENE}. In this section, we first introduce a few precise definitions describing the dense  regions phenomena in a bipartite graph $G$ and will use them to analyze various properties of our algorithm.

\begin{definition}[Isolated $\rho$-Synchronized Subgraph $\mathcal{G}(\mathcal{N}, \mathcal{M}, \rho)$]  \label{def.rho}
Given a bipartite graph $G = (V= N \cup M, E)$. Let $\mathcal{N} \subset N$ and $\mathcal{M} \subset M$. The subgraph induced by $(\mathcal{N}, \mathcal{M})$, denoted $\mathcal{G}(\mathcal{N}, \mathcal{M}, \rho)$, is isolated and $\rho$-synchronized if (1) $\forall m \in \mathcal{M}$, if there exists an edge $(c, m) \in E$, it must be that $c \in \mathcal{N}$; (2) $\forall n \in \mathcal{N}$, if there exists an edge $(n, c) \in E$, it must be that $c \in \mathcal{M}$; and (3) $\forall n \in \mathcal{N}, \exists \mathcal{W}  \subseteq \mathcal{M}$ such that $\forall m \in \mathcal{W}$ there exists an edges  $(n, m) \in E$. $\rho$, called the synchrony of the subgraph, is defined as
\vspace{-0.5em}
\begin{displaymath}
\rho = \frac{\overline{| \mathcal{W}|}}{|\mathcal{M}|},
\end{displaymath}
where $\overline{| \mathcal{W}|}$ is the mean for all $|\mathcal{W}|$s. 
\end{definition}
\vspace{-0.5em}

When $\rho = 1$, Definition \ref{def.rho} represents an extremely simple type of fraud pattern: fraudulent accounts and target nodes form a biclique that is disconnected from other parts of the network. Although naive, this is the core pattern that fraud detection algorithms try to capture \cite{CATCHSYNC,SPOKEN,NETPROBE,FRAUDAR}. To avoid detection, fraudsters often ``camouflage'' their activities by introducing additional edges in the graph \cite{FRAUDAR}, making the subgraph both less synchronized ($\rho < 1$ if all connected nodes are included) and not isolated. An effective detection algorithm must be robust against camouflaging and other perturbations.

Camouflaging can be classified into two types:
\vspace{-0.5em}
\begin{definition} [Active Camouflage (A-Cam)] \label{def:acam}
Given $\mathcal{G}(\mathcal{N}, \mathcal{M}, \rho)$, let $\hat{\mathcal{M}}$ be a set of target nodes and $ \hat{\mathcal{M}} \cap \mathcal{M} = \emptyset$. Then we use $\mathcal{G}(\mathcal{N}, \mathcal{M}, \rho)$+ A-Cam to denote the subgraph induced by $(\mathcal{N}, \mathcal{M})$ where $\forall m \in \hat{\mathcal{M}}$, there exists an edge $(n, m) \in E$, $n \in \mathcal{N}$. 
\end{definition}
\vspace{-0.5em}
In practice, fraudsters can easily add edges from $\mathcal{N}$ to $\mathcal{M}$ as well as to target nodes outside $\mathcal{M}$, since they control the fraud source nodes in $\mathcal{N}$ (Definition \ref{def:acam}). This pattern is frequently observed in social networks \cite{FRAUDAR,CATCHSYNC}. 

 \vspace{-0.5em}
 \begin{definition} [Passive Camouflage (P-Cam)] \label{def:pcam}
 Given $\mathcal{G}(\mathcal{N}, \mathcal{M}, \rho)$, let $\hat{\mathcal{N}}$ be a set of source nodes and $ \hat{\mathcal{N}} \cap \mathcal{N} = \emptyset$. Then we use $\mathcal{G}(\mathcal{N}, \mathcal{M}, \rho)$+ P-Cam to denote the subgraph induced by $(\mathcal{N}, \mathcal{M})$ where $\forall n \in \hat{\mathcal{N}}$, there exists an edge $(n, m) \in E$, $m \in \mathcal{M}$. 
 \end{definition}
\vspace{-0.5em}

P-Cam models the situations where some legit users add edges to fraudulent nodes in $\mathcal{M}$ by accident or share some resources (e.g., IP addresses, devices, etc.) with fraud groups. It is not very frequent, since legit users are out of fraudster's control, but is indeed possible. For example, we have observed that on a popular Chinese microblogging website, the system sometimes make users involuntarily follow some accounts. Also, explicitly modeling P-Cam makes it clearer to analyze the problem and algorithms.  

\subsection{A Biclique Problem with Linear Solution}

The definitions introduced earlier are closely related to the concept of biclique in bipartite graphs. A biclique is a complete sub-bipartite graph that contains all permissible edges, which, according to definition \ref{def.rho}, is essentially a subgraph with synchrony $\rho = 1$ (but not necessarily isolated). A biclique is said to be maximal if it is not contained in any other bicliques.  
The \emph{vertex maximum biclique problem} and the \emph{edge maximum biclique problem} are two distinct well-known problems in bipartite graphs. The former can be solved in polynomial time \cite{book1980michael}, while the latter is NP-complete\cite{biclique}.

Finding all \emph{maximal bicliques} in bipartite graphs, the maximal biclique enumeration problem, cannot be solved in polynomial time \cite{eppstein1994arboricity}, because it contains all edge maximum and vertex maximum bicliques. Detecting them would be intractable. Inspired by real-world fraud attack patterns, we introduce a novel, more restricted biclique problem that captures the essence of group synchrony in the fraud activities yet admits a \emph{linear} algorithm. A maximal biclique has no constraint on how it is connected to the other nodes of the graph. Now we introduce the notion of \emph{half isolated biclique} that only allows edges to outside nodes from at most \emph{one} part of the subgraph. Formally,


\vspace{-0.5em}
\begin{definition}  [Half Isolated Biclique (HI Biclique)] \label{def:HI}
In a bipartite graph $G = (N \cup M, E)$, a subgraph  $\mathcal{G} = (\mathcal{N \cup M, E})$ is a half isolated biclique if (1) $\mathcal{G}$ is a biclique $(\forall n \in \mathcal{N}, m \in \mathcal{M}, \ (n, m) \in E)$; (2) \{ $\forall (m \in M) \wedge (m \notin \mathcal{M}),  n \in \mathcal{N}, then \ (n, m) \notin E$ \} or \{ $\forall (n \in N) \wedge (n \notin \mathcal{N}), m \in \mathcal{M},  then \ (n, m) \notin E$ \}.      
 \end{definition}
 \vspace{-0.5em}

Based on this definition, we propose the notion of \emph{maximal half isolated biclique}.

\begin{definition}  [Maximal Half Isolated Biclique (MHI Biclique)] \label{def:MHI}
In a bipartite graph $G = (N \cup M, E)$, a subgraph  $\mathcal{G} = (\mathcal{N \cup M, E})$ is a maximal half isolated biclique if (1) $\mathcal{G}$ is a HI biclique; (2) there does not exist a HI biclique $\mathcal{G}_1$ = $(\mathcal{N}_1 \cup \mathcal{M}_1, \mathcal{E})$ where $\mathcal{N}_1 \supseteq \mathcal{N}$ and $\mathcal{M}_1 \supseteq \mathcal{M}$.      
 \end{definition}

An MHI biclique is essentially a biclique with at least one of its parts ($\mathcal{N}$ or $\mathcal{M}$) isolated. Even though restricted, 
MHI bicliques still cover a lot of interesting scenarios: the set of MHI bicliques in $G$ consists of all maximal isolated bicliques ($\mathcal{G}(\mathcal{N}, \mathcal{M}, \rho =1.0)$ if $\mathcal{G}$ is a maximal biclique), all maximal isolated bicliques + A-Cam, and all maximal isolated bicliques + P-Cam. Many real-world phenomena manifest themselves as MHI bicliques. As mentioned before, fraud groups with A-Cam are likely to form MHI bicliques. MHI bicliques can also be found in many other applications such as retweet boosting detection\cite{CATCHSYNC,FRAUDAR}, network intrusion\cite{MZOOM},  and genetic applications\cite{GENE,sanderson2003obtaining,mushlin2007graph}. Thus the following problem has a wide range of applications: 

\medskip
\noindent \textbf{Maximal Half Isolated Biclique Enumeration Problem (MHIBP):} In a bipartite graph $G = (N \cup M, E)$, enumerate all maximal half isolated bicliques.
\\

Solving MHIBP indicates finding a special set of  maximal bicliques. This may have considerable significance in graph theory and related areas. To the best of our knowledge, ours is the first linear-time complexity solution. This could be of independent interest.

In summary, we cast the group fraud detection problem as discovering a special case of bicliques, MHI Bicliques. The restricted nature of MHI Biclique avoids the intractability of the standard biclique problems. In the rest of the paper, we will introduce a linear time algorithm for MHIBP. We also provide a few theoretical results and empirical study regarding the excellent effectiveness and robustness of the algorithm, especially when applied to the application of anti-fraud.

\section{Suspiciousness Tree}\label{stree}

In this section, we introduce a novel data structure, denoted by Suspiciousness Tree or S-Tree, for solving MHIBP and variants. The idea is inspired by frequent-pattern tree (FP-tree) in association rule mining\cite{FPTREE}. To make it concrete, we describe the scheme in the context of fraud detection which requires a special score, denoted F-score that will be introduced later. However, we stress that F-score is only essential for fraud detection. It is easy to verify that, using F-score directly for ordering nodes, or replacing it with frequency or simply dictionary ranking of the node's identifier, or any consistent nodes ordering mechanism can solve the general MHIBP.

\subsection{Constructing Baskets}  

The construction of S-tree starts with modeling each target nodes with a structure called \emph{basket}, which captures a target node's local connectivity. Formally, given a target node $m \in M$, the basket $B(m)$ of $m$ is constructed as:
\begin{equation}
B(m) = \lbrace m, \ I(m), \ f(m)  \rbrace
\end{equation}
where $m$ is the identifier of $B(m)$. $I(m) = \{n \in N : (n, m) \in E\}$ is the set of source nodes that $m$ is connected with. $f(m)$, called the F-score, is a preliminary estimate of suspiciousness of $B(m)$, using only information local to $m$. 
F-score can be thought of as indicating how suspicious that source nodes in $I(m)$ appear in $B(m)$ when the basket is examined in isolation. The final suspiciousness of a source node is determined globally by both the nodes' F-scores and the structure of the graph.
In the following, we will first introduce our choices of F-score that work well empirically. We then present the mining algorithm that detects subgraphs with high synchrony.


%


\para{Determining $f(m)$.} Recall that we distinguish two modes: AOBG and ARBG where the bipartite graphs model different relations. ''Unusualness'' is different in the two modes.

\emph{AOBG mode.}  When elements in $M$ are \emph{objects}, an edge represents an interaction between an account and an object (e.g., a product or a followee). Intuitively, a target node with higher in-degree are less suspicious, since it has no incentive to collude with fraudster for the acquisition of popularity (e.g., a celebrity on Twitter). In this case, we define $f(m)$ as
\begin{equation}
f(m) = log(\frac{|E|}{|I(m)|+c}),
\end{equation}
where $c$ is a small constant to prevent excessive variability for small values of $|I(m)|$. Note that Max($|I(m)|) + c < |E|$. 

This can also been thought of as the self-information or surprisal \cite{INFORMATION} of a randomly chosen edge falls on a certain the basket. 

\emph{ARBG mode.}  In contrast, when nodes in $M$ represent resources (e.g., IP addresses, devices, etc.) that the accounts in $N$ use when interacting with the platform, a target node with higher in-degree is more suspicious as resource sharing is a key characteristic of group fraud. For example, it is highly suspicious for more than 1k users to login on the same IP address. In fact many works (e.g.,   \cite{SYNCHROTRAP,COPYCATCH,FRAUDAR,CROSSSPOT}) rely on detecting such sharing. In this case, we define $f(m)$ as
\begin{equation}
f(m) = log(|I(m)|+c).
\end{equation}
\smallskip

By the transformation, we acquire the baskets $\mathbf{B} = \{B(m), \forall m \in M\}$ from $G$. And we use $\mathbf{B}(\mathcal{M})$ to denote the set $\{B(m), \forall m \in \mathcal{M}\}$, and $\mathbf{I}(\mathcal{M})$ to denote the set $\{I(m), \forall m \in \mathcal{M}\}$.



\para{Time Complexity.} All baskets $\mathbf{B}$ can be constructed by one scan of $E$. Therefore, the time complexity of constructing $\mathbf{B}$ is $O(|E|)$. Noticing that $\sum_{m \in M} |I(m)| = |E|$. 

\begin{figure}[t]
\includegraphics[width = 0.8 \columnwidth, height = 0.6 \columnwidth ]{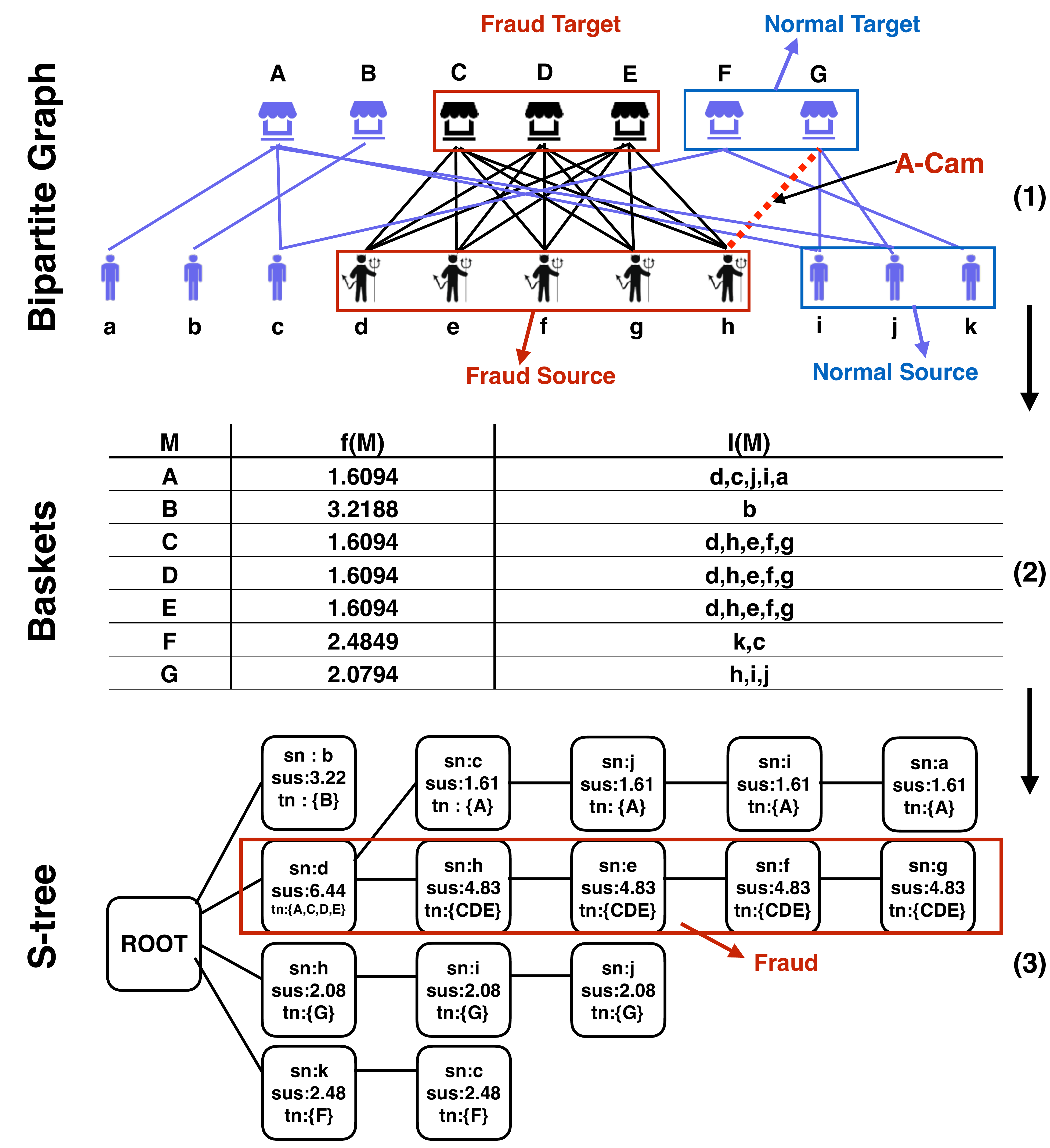}
\centering
\caption{Toy example of S-tree (Overflow) } \label{fig:examp}
\vspace{-2em}
\end{figure}

\subsection{Building Suspiciousness Tree}

The basket captures the local connectivity of a target node. Given a high-synchrony subgraph 
$\mathcal{G}(\mathcal{N}, \mathcal{M}, \rho \rightarrow 1.0 )$, nodes within $\mathcal{M}$ have similar neighbor sets. This can be shown with a toy example in Fig.\ref{fig:examp}. The fraud group consisting of a set of target nodes $\{C, D, E \}$ and a set of source nodes $\{d, e, f, g, h\}$, forming a dense subgraph. After baskets construction, $C, D, E$ result in identical neighbor sets: $I(C) = I(D) = I(E)$. Detecting high-synchrony subgraphs can be done efficiently via clustering similar nodes together. 

In addition to subgraphs with high synchrony, we also focus on \emph{suspicious} dense subgraphs. Consider a biclique formed by $(\mathcal{N}, \mathcal{M})$ and we assume $\mathcal{N}$ is a set of followers on Twitter. Intuitively, $(\mathcal{N}, \mathcal{M})$ is more suspicious if $\mathcal{M}$ is a set of ``nobody'' instead of a group of celebrities. Thus, we propose \emph{suspiciousness tree} or S-tree. Compared to FP-tree\cite{FPTREE}, S-tree has novel node structure, sorting metric, and increment method when its node is shared.

%

The construction of a S-tree is a process of handling each basket in $\mathbf{B}$. Let $T$ be the S-tree. A node $x$ on $T$ contains three fields: \textbf{sn}, \textbf{sus}, and \textbf{tn},  where $x.sn$ represents a particular node of $N$ or $M$ , $x.sus$ records a suspiciousness score for $x$, and $x.tn$ denotes a subset of associated nodes of $M$ or $N$.
 
Within a basket, the source nodes are sorted by their total suspiciousness  scores $g$ defined as 

\begin{equation}
g(n) = \sum_{m \in H(n)}  f(m),
\end{equation}
where $H(n) = \{m \in M:(n,m) \in E\}$ is set of target nodes that $n$ has edge with. 
Metric of this from obeys two basic properties (or `axioms'):
\begin{enumerate}
\item All other conditions being equal, the more frequent $n$ is in $\mathbf{B}$, the greater $g(n)$: $|H(n)| \uparrow \Rightarrow  g(n) \uparrow$.
\item All other conditions being equal, the more suspicious a basket that $n$ occurs in is,  the greater $g(n)$: $f(m) \uparrow \Rightarrow  g(n) \uparrow$.
\end{enumerate}

%

\begin{algorithm}[h]  
\caption{ S-tree construction } \label{alg.stree}
\begin{algorithmic}[1] 
\Require  $\mathbf{B}$, $g$
\Ensure $T$
\For{each $B(m) \in \mathbf{B}$}
\State sort $I(m)$ in a descending order of $g$ (Eq.(4))
\EndFor
\State $R \leftarrow  T$.root
\For{each $B(m) \in \mathbf{B}$}
\State $x \leftarrow R$
\While {$I(m) \neq \emptyset$}
\State  $n \leftarrow I(m).top()$  \# remove and return the first node in $I(m)$
\If{$x$ has a child $\hat{x}$ and $\hat{x}.sn = n$ }
\State  $\hat{x}.sus \leftarrow  \hat{x}.sus + f(m)$
\State $\hat{x}.tn \leftarrow  \hat{x}.tn \cup \{m\}$
\Else
\State $\hat{x} \leftarrow$ new node
\State $\hat{x}.sn \leftarrow n, \ \hat{x}.sus \leftarrow f(m),$ \ $\hat{x}.tn \leftarrow \{m\}$
\State $x$ add a child $\hat{x}$
\EndIf
\State $x \leftarrow \hat{x}$
\EndWhile
\EndFor
\State \Return $T$ 
\end{algorithmic}  
\label{alg:greedy}
\end{algorithm}

Given $\mathbf{B}$, S-tree $T$ is built by Algorithm \ref{alg.stree}.
Before inserting $\mathbf{B}$ into $T$, for each $B(m) \in \mathbf{B}$, we sort $I(m)$ in a descending order of $g$, which can provide better chances that the most suspicious and frequent source nodes are ranked at the head (line 1-2). Then, the construction begins from the root node that has no children (line 3). We process each basket one by one (line 4-15). In the process of processing one basket, we insert each source node within the basket in order (line 6-15). If a node to be inserted has already been shared on the path of the tree, its suspiciousness and neighboring nodes will be accumulated (Line 8-10). Otherwise, a new node will be created (line 11-14).    


\begin{definition} [Path $p(x)$] Given an S-tree $T$, let $x$ be a node of $T$, then $p(x)$ denotes the set of nodes along the path from the root node to $x$ in $T$, with the root node excluded.
\end{definition}

\para{Observation 1 : dense subgraphs $\Rightarrow$  shared path.} 
Consider a dense subgraph $\mathcal{G}(\mathcal{N}, \mathcal{M}, \rho \rightarrow 1.0)$. Then the baskets $\mathbf{B}(\mathcal{M})$ are highly similar to each other. By sorting  $\mathbf{I}(\mathcal{M})$ in terms of $g$, there are better chances that $\mathbf{I}(\mathcal{M})$ share a common prefix, denoted by $\mathcal{N}$. Then $\mathbf{B}(M)$ will be mapped into a subtree of S-tree, where $\forall x \in p(\hat{x}), x.sn \in \mathcal{N}, x.sus = \sum_{m \in \mathcal{M}}f(m),$ and $x.tn = \mathcal{M}$, if $\hat{x}.sn$ is the last node of sorted $\mathcal{N}$ by $g$.

It is easy to verify that $T$ has anti-monotonicity property:

\begin{theorem}\label{theo:anti}
[The Anti-monotonicity Property]
Let $x$ be a node of $T$. Then $\forall c \in p(x) $, it must be that $c.sus \geq x.sus$ and $c.tn \supseteq x.tn$.
\end{theorem}
\vspace{-0.5em}

\begin{proof}
The proof is obvious. When processing a basket $B(m)$ in Algorithm 1, and when a node $n \in I(m)$ is added to the tree, all its parent nodes' \emph{sus} fields will be incremented by $f(m)$ and their \emph{tn} fields will be expended by $\{m\}$.
\end{proof} \smallskip

\para{Time complexity.} From the S-tree construction, the cost of inserting a basket $B(m)$ into S-tree is $O(|I(m)|)$. Thus, the construction of S-tree takes $O(|E|)$ by inserting $\mathbf{B}$.

\para{Space complexity}
As a basket corresponds to a particular path of S-tree,  the size of S-tree is upper bounded by $|E|$. It reaches the bound when no node-sharing happens in the S-tree. And the depth of S-tree is bounded by Max$(|I(m)|)$.

\section{Dense Subgraph Mining with S-tree}

Building upon the S-tree structure, we now introduce two algorithms for dense subgraph mining. The former solves the generic MHIBP, while the latter is designed specifically for fraud group detection.

\subsection{Algorithm for Solving MHIBP}

Consider a maximal isolated biclique (MI biclique) and we use MI biclique + A-Cam to denote an isolated biclique + A-Cam  that is not contained in other isolated bicliques + A-Cam. We recall that MHI bicliques contain MI bicliques, MI bicliques + A-Cam, and MI bicliques + P-Cam. First, let us see how MHI biclique to be represented on S-tree. 

A \emph{branch} node is a node with more than one children. A \emph{leaf} node is a node without any child. We define a special type of nodes as the following: 

\begin{definition} [Narrow Node] \label{def:narrow}
In an S-tree $T$, let $x$ be a node that is neither a leaf node nor branch node. And let 
 $c$ denote the only child of $x$. Then $x$ is a narrow node if $x.tn \supset c.tn$.
\end{definition}

Then, we give the definition of \emph{maximal-shared path}:

\begin{definition} 
[Maximal-Shared Path]
Given a path $p(x)$, then we call $p(x)$ is a maximal-shared path if $x$ is a leaf, branch or narrow node. 
\end{definition}

Based on the two definitions,  we have the following Theorems. For brevity, we use $[\mathcal{N}, \mathcal{M}]$ to denote a biclqiue induced by $\mathcal{N}$ and $\mathcal{M}$.

\begin{theorem}\label{theo:sharing2}
Given an MI biclqiue $\mathcal{G}(\mathcal{N}, \mathcal{M}, \rho = 1.0)$, then  $\mathcal{G}$ must form a path $p(x)$ where $x$ is a leaf node in $T$.
\end{theorem}
\vspace{-0.5em}

\begin{proof}
Because  $\mathcal{G}(\mathcal{N}, \mathcal{M}, \rho = 1.0)$ is isolated, for each $n \in \mathcal{N}$, $n$ only occurs in $\mathbf{B}(\mathcal{M})$. 
Because $\rho = 1.0$, there must be that $\forall m_1, m_2 \in \mathcal{M}$, $I(m_1) = I(m_2)$, based on the basket construction. By Algorithm 1, 
 $\mathbf{I}(\mathcal{M})$ must have same order and $\mathbf{B}(\mathcal{M})$ must form a shared path $p(x)$ where $x$ is a leaf node in $T$ and it must hold that $x.tn = \mathcal{M}, x.sus = \sum_{m \in \mathcal{M}}f(m)$, and $\{c.sn, \forall c  \in p(x) \} = \mathcal{N}$.
\end{proof}

\vspace{-0.5em}
\begin{theorem}\label{theo:sharing1}
Let $T$ be the S-tree for $M$.
Given an MI biclqiue $\mathcal{G}(\mathcal{N}, \mathcal{M}, \rho = 1.0)$ (+ A-Cam), then  $\mathcal{G}$ must form a maximal-shared path $p(x)$ in $T$. Let $\mathcal{M}$ be the set $\bigcup_{c \in C}c.tn$ where $C$ denotes the children of $x$, $\mathcal{N}$ be the set $\{c.sn, \forall c  \in p(x) \}$. Then, $[\mathcal{N},  x.tn]$ must be an MI bilique (+ A-Cam) if $x$ is a leaf node and $ [ \mathcal{N},  (x.tn - \mathcal{M})] $ must be an MI bilique + A-Cam if $x$ is either a branch node or narrow node and $(x.tn - \mathcal{M}) \neq \emptyset$.
\end{theorem}

\begin{proof} 
A-Cam only introduces an edge pointed from a source node $n \in \mathcal{N}$ to a target node $m \not \in \mathcal{M}$, thus the camouflage does not change baskets $\mathbf{B}(\mathcal{M})$. 
Because $\rho = 1.0$ and Theorem \ref{theo:sharing2}, $\mathbf{B}(\mathcal{M})$ must form a shared path $p(x)$ where $x$ is last element of $I(m), m \in \mathcal{M}$. Now, let us consider the position of $x$ in $T$. 

First, suppose $x$ is \emph{not} a branch, leaf, or narrow node. Then $x$ must have exactly one child node $c$ and $c.tn = x.tn$ based on Definition \ref{def:narrow} and Theorem \ref{theo:anti}. Thus $p(c)$ must be mapped by $[ \{c_2.sn, \forall c_2  \in p(c) \},  c.tn]$.   Because $p(x) \subset p(c)$, we conclude that 
$
[ \{c_1.sn, \forall c_1  \in p(x) \},  x.tn]  \subset [ \{c_2.sn, \forall c_2  \in p(c) \},  c.tn]$,
which contradicts the definition of MI biclique(+ A-Cam).

Therefore $x$ must be either a branch, leaf, or narrow node. 
If $x$ is a
\begin{enumerate}
\item branch node: $[\mathcal{N},  x.tn]$ must not be an MI biclique + A-Cam, because for each $m_1 \in x.tn$, $ \exists m_2 \in x.tn$, then $I(m_1) \neq I(m_2)$ due to  $|C|>1$,  which is contradicted to Theorem \ref{theo:sharing2}. If $|x.tn| > |\mathcal{M}|$, then $[\mathcal{N},  \mathcal{M}]$ must not be an MI biclique + A-Cam, because for each $m_1 \in \mathcal{M}$, $ \exists n_1 \in I(m_1)$ satisfying $n_1 \not \in \mathcal{N}$,  which is contradicted to Definition \ref{def:acam}. And, 
 $ [\mathcal{N},  (x.tn - \mathcal{M})] $ must be an MI biclique + A-Cam, because 
$\forall m \in (x.tn - \mathcal{M})$, $I(m)$ must be same and it is not contained in other MI biclique + A-Cam.   
\item leaf node: Based on Theorem \ref{theo:sharing2}, $[ \mathcal{N},  x.tn]$ must be an MI biclique (+ A-Cam).
\item narrow node: $x$ has one child $c$, $x.tn \supset c.tn$, and $\mathcal{M}$ is $c.tn$. Then $[ \mathcal{N},  (x.tn - \mathcal{M})]$ must be an MI biclique + A-Cam (proof is same to the condition if $x$ is a branch node).
\end{enumerate}
Therefore MI biclique(+ A-Cam) must form a maximal-shared path and for each maximal-shared path $p(x)$, $[\mathcal{N},  x.tn]$ must be an MI bilique (+ A-Cam) if $x$ is a leaf node and $ [ \mathcal{N},  (x.tn - \mathcal{M})] $ must be an MI bilique + A-Cam if $x$ is either a branch node or narrow node and $(x.tn - \mathcal{M}) \neq \emptyset$.
\end{proof}

As previously mentioned, $T$ is the S-tree for the target nodes in $G$. We can also build the S-tree for the source nodes using the same methods. We denote such an S-tree by $T_1$. 

\begin{theorem}\label{theo:sharing3} 
Let $T_1$ be the S-tree for $N$. Given an MI biclqiue $\mathcal{G}(\mathcal{N}, \mathcal{M}, \rho = 1.0)$ (+ P-Cam), then  $\mathcal{G}$ must form a maximal-shared path $p(x)$ in $T_1$. Let $\mathcal{M}$ be the set $\bigcup_{c \in C}c.tn$ where $C$ denotes the children of $x$, $\mathcal{N}$ be the set $\{c.sn, \forall c  \in p(x) \}$. Then, $[\mathcal{N},  x.tn]$ must be an MI bilique (+ P-Cam) if $x$ is a leaf node and $ [ \mathcal{N},  (x.tn - \mathcal{M})] $ must be an MI bilique + P-Cam if $x$ is either a branch node or narrow node and $(x.tn - \mathcal{M}) \neq \emptyset$.
\end{theorem}
\vspace{-0.5em}
\vspace{-0.5em}
\begin{proof}
The proof is same as the proof in Theorem \ref{theo:sharing1}.
\end{proof}
\vspace{-0.5em}

Let $\mathcal{G}(\mathcal{N}, \mathcal{M}, \rho = 1.0)$ denote an MI biclqiue. Based on Theorem 2-4, we propose Algorithm 2-3 to catch all $\mathcal{G}$ (+ A-Cam) in $T$ or $\mathcal{G}$ (+ P-Cam) in $T_1$.

\begin{algorithm}[ht]
\caption{\emph{main}} 
\begin{algorithmic}[1]
\Require  $T$
\Ensure $\hat{\mathcal{G}}s$
\State $\hat{\mathcal{G}}s \leftarrow \emptyset$
\State $R \leftarrow T.root$
\State \emph{find MHI bicliques}($R, \hat{\mathcal{G}}s$)
\State \Return $\hat{\mathcal{G}}s$
\end{algorithmic} 
\end{algorithm}
\vspace{-0.5em}

\begin{algorithm}[h]
\caption{\emph{find MHI bicliques}} 
\begin{algorithmic}[1]
\Require  $x$, $\hat{\mathcal{G}}s$
\State  $C \leftarrow x.childern$ 
\State $\mathcal{M} \leftarrow x.tn$
\State $A \leftarrow \{a.xn, \forall a \in p(x) \}$
\If{$|C| > 1$}
   \State $\mathcal{M}' \leftarrow \emptyset$
    \For {each $c \in C$}
    \State $\mathcal{M}' \leftarrow \mathcal{M}' \cup c.tn$
    \EndFor
    \If {$|\mathcal{M}| > |\mathcal{M}'|$}
    \State $\hat{\mathcal{G}}s \leftarrow  \hat{\mathcal{G}}s \cup \{ (A, \mathcal{M} - \mathcal{M}') \}$ 
    \EndIf
\EndIf
\If{$|C| = 1$}
    \State $\mathcal{M}' \leftarrow x.child.tn$
    \If{$|\mathcal{M}| > |\mathcal{M}'|$}
    \State $\hat{\mathcal{G}}s \leftarrow  \hat{\mathcal{G}}s \cup  \{(A, \mathcal{M} - \mathcal{M}')\}$ 
    \EndIf
\EndIf
\If{$|C| = 0$}
     \State $\hat{\mathcal{G}}s \leftarrow  \hat{\mathcal{G}}s \cup \{(A, \mathcal{M})\}$
\State \Return
\EndIf
\For{each $\hat{x} \in C$}
\State \emph{find MHI bicliques}($\hat{x}, \hat{\mathcal{G}}s$)
\EndFor
\end{algorithmic}  
\end{algorithm}
\vspace{-0.5em}

Obviously, Algorithm 3 is a depth-first search algorithm. It searches all leaf nodes (line 14-15), branch nodes (line 4-9), and narrow nodes (line 10-13) by one scan all nodes. Thus Algorithm 2-3 catch all maximal-shared paths and corresponding $\mathcal{G}$ ( + A-Cam).

\begin{theorem}
Let $\hat{\mathcal{G}}_{S}$s be the bicliques returned by running Algorithm 2-3 on $T$, and $\hat{\mathcal{G}}_{T}$s be the bicliques returned by running Algorithm 2-3 on $T_1$.
Then $\hat{\mathcal{G}}$s $= merge ( \hat{\mathcal{G}}_{S}$s, $\hat{\mathcal{G}}_{T}$s) (Algorithm \ref{alg:merge}) must be all MHI bicliques in $G$. 
\end{theorem}

\begin{proof}
Based Theorem \ref{theo:sharing1} and Theorem \ref{theo:sharing2}, running Algorithm 2-3 on $T$ detects all MI biclqiues (+ A-Cam), while running Algorithm 2-3 on $T_1$ finds all MI biclqiues (+ P-Cam).  And \emph{merge}($\hat{\mathcal{G}}_{S}$s, $\hat{\mathcal{G}}_{T}$s) deletes all repeated and contained bicliques.
Thus Algorithm 2-4 + S-tree solves MHIBP.
\end{proof}

\begin{algorithm}[h] 
\caption{\emph{merge}} \label{alg:merge}
\begin{algorithmic}[1]
\Require  $\hat{\mathcal{G}}_{S}$s, $\hat{\mathcal{G}}_{T}$s 
\For {each $\hat{\mathcal{G}}$s $\in \{ \hat{\mathcal{G}}_{S}$s, $\hat{\mathcal{G}}_{T}$s\}}
\For {each $\hat{\mathcal{G}} \in \hat{\mathcal{G}}$s} 
\State $(\mathcal{N, M}) \leftarrow \hat{\mathcal{G}}$ 
\State $n \leftarrow \mathcal{N}.top()$ \# Return a node of $\mathcal{N}$.
\State let $p(x)$ be the path mapped by $B(n)$ in $\hat{T}$. \# If $\hat{\mathcal{G}} \in \hat{\mathcal{G}}_{S}$s, $\hat{T} \leftarrow T_1$; Else $\hat{T} \leftarrow T$ 
\For {each $c \in p(x)$}
\If { $c$ is a leaf, branch, or narrow node}
    \State $\mathcal{G}'  \leftarrow$ biclique represented by $c$ \# See Theorem 3-4
    \If{$ \hat{\mathcal{G}} \subseteq  \mathcal{G}'$}
    \State remove $\hat{\mathcal{G}}$ from $\hat{\mathcal{G}}$s
    \EndIf 
\EndIf
\EndFor
\EndFor
\EndFor
\State \Return $\hat{\mathcal{G}}_{S}$s + $\hat{\mathcal{G}}_{T}$s
\end{algorithmic}  
\end{algorithm}

For each $ \hat{\mathcal{G}} \in \hat{\mathcal{G}}_{S}$s, there is a guarantee that $\hat{\mathcal{G}}$ must be an MI biclqiue (+ A-Cam). However,  $\hat{\mathcal{G}}$ may be contained in an MI biclqiue + P-Cam.  Algorithm \ref{alg:merge} describes the way we merge repeated or contained HI bicliques.
Algorithm \ref{alg:merge} efficiently prunes unnecessary pairwise comparisons, where a $\hat{\mathcal{G}}$ is only compared with a subset of $\hat{\mathcal{G}}_{T}$s.
If $\hat{\mathcal{G}} \subseteq  \mathcal{G}'$, $\mathcal{G}'$ must 
share each source node (eg., $n$) in $\hat{\mathcal{G}}$ and $\mathcal{G}'$ must be represented by either a  leaf, branch, or narrow node on the shared path mapped by $B(n)$ on $T_1$ (line 3-8).

\subsection{Practical Algorithm for Fraud Group Detection} \label{sec:mining}

In the previous section we provide theoretical results on MHIBP which essentially deals with subgraphs with full synchrony ($\rho = 1$). In the real world, our target pattern (e.g., fraud groups) may not form a strict biclique, but a dense subgraph. Let $\mathcal{G}(\mathcal{N}, \mathcal{M}, \rho \rightarrow 1.0)$ represent the subgraph formed by a fraud group and $t$ denote the subtree mapped by $\mathcal{G}$.   Based on Observation 1,  there will be a shared prefix $p(x)$ in $t$. Thus detecting $\mathcal{G}$ is equivalent to detecting $p(x)$. Building upon this property, we now introduce an algorithm that effectively detects isolated-dense subgraphs with A-Cam and/or P-Cam.

First, we set up two parameters, \emph{thickness} and \emph{depth}, which define a suspiciousness boundary: a node $x$ is suspicious if the depth of $x$ in $T$ is more than \emph{depth} and $x.sus$ is greater than \emph{thickness}. The boundary selects the subtrees with shared prefix that correspond to dense subgraphs. The algorithm is described in Algorithm 4. We use $x.descendant$ to denote the set $\{c \ $in \ $ T: x \in p(c)\}$.

\begin{algorithm}[h]
\caption{} \label{alg:p}
\begin{algorithmic}[1]
\Require  $T$, \emph{thickness}, \emph{depth}
\State $X_{sus} \leftarrow \emptyset$
\State Start search from the root node $R$.
\State Find all nodes $X$ of which depth $=$ \emph{depth}.
\For {$x$ in $X$}
    \If {$x.sus$ $\geq$ \emph{thickness}}
    \State $X_{sus} \leftarrow X_{sus} + p(x) + x.descendant$
    \EndIf
\EndFor
\State \Return  $X_{sus}$
\end{algorithmic}  
\end{algorithm}
\vspace{-0.5em}

In Algorithm \ref{alg:p}, we retain the whole subtree (line 6) if its shared prefix exceeds the threshold. This, however, may result in false alarms (legit users being included in $X_{sus}$). We propose the following metric to mitigate the issue. 

Given $X_{sus}$, we calculate a suspiciousness score (s-score) of a source node as:

\begin{equation} \label{eqn.s-score}
s(n) = \sum_{\forall x \in X_{sus}} x.sus \ \ \text{if} \ x.sn == n.
\end{equation}
We sort all nodes in $N$ in descending order of $s$, given $X_{sus}$.
S-score has a very nice property that helps suppressing false positives:  $B(m_i)$ may contain normal source nodes if some normal users accidentally create edges pointed to $m_i$. Given a set of fraud source nodes $\mathcal{N}$ and a set of fraud target nodes $\mathcal{M}$, $m_i \in \mathcal{M}$, let $n_g$ denote a normal source node having an edge with $m_i$. Because $n_g$ should not have edges with the majority of $\mathcal{M}$ while $\mathcal{N}$ does, i.e., $n_g$ only occurs in $B(m)$ while $\mathcal{N}$ frequently occurs in $\mathbf{B}(\mathcal{M})$.    Thus, $s(n_g)$ is much lower than $s(\mathcal{N})$.

\para{Determining \emph{thickness} and \emph{depth}.}
The thresholds \emph{thickness} and \emph{depth} can be determined empirically. In practice, we found that choosing the averages works well:
\begin{displaymath}
thickness = \frac{ \sum_{\forall x  \in T} x.sus}{|T|},
\end{displaymath}
\vspace{-0.5em}
And,

\begin{displaymath}
depth = \frac{|E|-|T|}{|\mathbf{B}|} ,
\end{displaymath}
where $|T|$ denotes the number of nodes of $T$ and $|E| = \sum_{B(m) \in \mathbf(B)} |I(m)|$.  ($|E|-|T|$) denotes the number of source nodes that are compacted to shared nodes in $T$, and  the number of paths in $T$ is upper bounded by $|\mathbf{B}|$. Thus, \emph{thickness} is average suspiciousness score of each node in S-tree. \emph{depth} is the average length of shared prefix of each subtree.

\vspace{-0.5em}
\section{Suspiciousness Forest} \label{sec:forest}
\vspace{-0.5em}

In this section, we extend S-tree into S-forest to support multimodal data.
\vspace{-0.5em}
\subsection{$1+ K$ Dataset}
\vspace{-0.5em}

S-forest handles what we call $1+K$ datasets. A $1+K$ dataset $D(X, A_1,..., A_K)$ is a collection of entries each of which has $1+K$ fields. The first field, denoted $X$, is an identifier for that entry, often representing the entity that we are trying to classify. The rest $K$ fields, denoted by $\{ A_1,..., A_K\}$, are attributes or features of this entity. The $1+ K$ formulation is a very common data model applicable to any scenario where a sample can be represented by an identifier and a number of features. For example, in Twitter, a follow action can be represented $1+3$ by $X = follower$ and 3 dimensional features $\{followee, IP \ address, timestamp\}$.
\vspace{-0.5em}
\subsection{Build Suspiciousness Forest}
\vspace{-0.5em}

Given $D(X, A_1,..., A_K)$, for $k \in \{1,...,K\}$, a bipartite graph $G_k$ is formed using values of $X$ and $A_k$. An S-tree  $T_k$ is then built on $G_k$ using same method in Sec.4. Thus we can obtain the S-forest: $\mathbf{T} = \lbrace T_1 , ... T_k, ... T_K \rbrace$.

S-forest is a natural extension of S-tree into multidimensional data. Note that the features $A_k$ can be either entities or resources so the AOBG and ARBG modes we discussed earlier are all special cases of S-forest (with $K = 1$). By looking at $1+K$ data, we can not only combine information and modeling power from both AOBG and ARBG, but also utilize more information from multiple dimensions.

For the purpose of detecting suspicious entities in $X$, we propose the following suspiciousness score acquired from $\mathbf{T}$, which is simply a weighted sum of suspiciousness scores by each individual tree: 

\vspace{-0.5em}
\begin{equation}
S(n) = \sum_{k=1}^K w_k s_k(n),
\end{equation}
where the $s_k(n)$ is the $s(n)$ of computed on $T_k$ by equation \ref{eqn.s-score}. $w_k$ is the weight of $T_k$. $w_k$'s can be automatically regressed when labeled data are available. Or they can be determined empirically. In our experiments, we found the following simple form works well: 
\begin{displaymath}
w_k = \log  q_k,
\end{displaymath}
where $q_k$ is the number of unique values of $A_k$. The intuition behind this choice is as follows. First, since each entry in $D$ corresponds to an edge in $G_k$, all the bipartite graphs have equal number of edges. Unique values of  $A_k$ correspond to target nodes in $G_k$. The larger $q_k$ is, the sparser $G_k$. Let $p_k = 1/q_k$ which can be interpreted as the probability that a random edge lands on any particular target node (i.e., value of $A_k$), assuming target nodes are selected according to uniform distribution. $w_k = \log  q_k = \log \frac{1}{p_k}$ is simply the self information or surprisal of such an event, which we use to represent how unusual or suspicious the graph is.

\para{Analysis.}
S-forest uses an additive model to ensemble multiple graphs built on different features. This property, together with S-tree's capability for detecting dense subgraphs on a single graph, makes S-forest superior in several aspects. We show two examples in the following, which will be validated by experiments in section \ref{exp:sforest}. 

\para{Property 1: Detecting Groups with Overlap.}
We have observed that when there are groups with overlapping nodes, existing approaches such as \cite{FRAUDAR, MZOOM, CROSSSPOT, DCUBE} fail to detect them while S-forest can. The reason is as follows. Fig.\ref{fig:group} represents a typical case of two overlapping groups. Two fraud groups $A$ and $B$ form two dense subgraphs on two graphs respectively. Let $\mathcal{W}$ be the set of common source nodes shared by $A$ and $B$. $\mathcal{W}$ will appear in both the subtrees formed by $A$ and $B$.  Let the S-scores of $A$ and $B$ be $S(A)$ and $S(B)$, respectively, then $S(\mathcal{W}) = S(A)+ S(B)$. The elevation of $S(\mathcal{W})$ does not decrease $S(A)$ or $S(B)$. With properly chosen threshold, all nodes in $A \cup B$ can be detected by S-forest.

Unfortunately, the approaches such as \cite{FRAUDAR, MZOOM, CROSSSPOT, DCUBE} that detect the dense blocks in tensors may not work. Lacking a way to aggregate or select key features, they resort to a greedy method for dense blocks detection. As a result, they can detect the densest block $\mathcal{W}$. However, once $\mathcal{W}$ is detected and removed, the remaining blocks become sparse and are likely to be missed.

\vspace{-0.5em}
\begin{figure}[h]
\includegraphics[width = 0.8\columnwidth,height = 0.3\columnwidth]{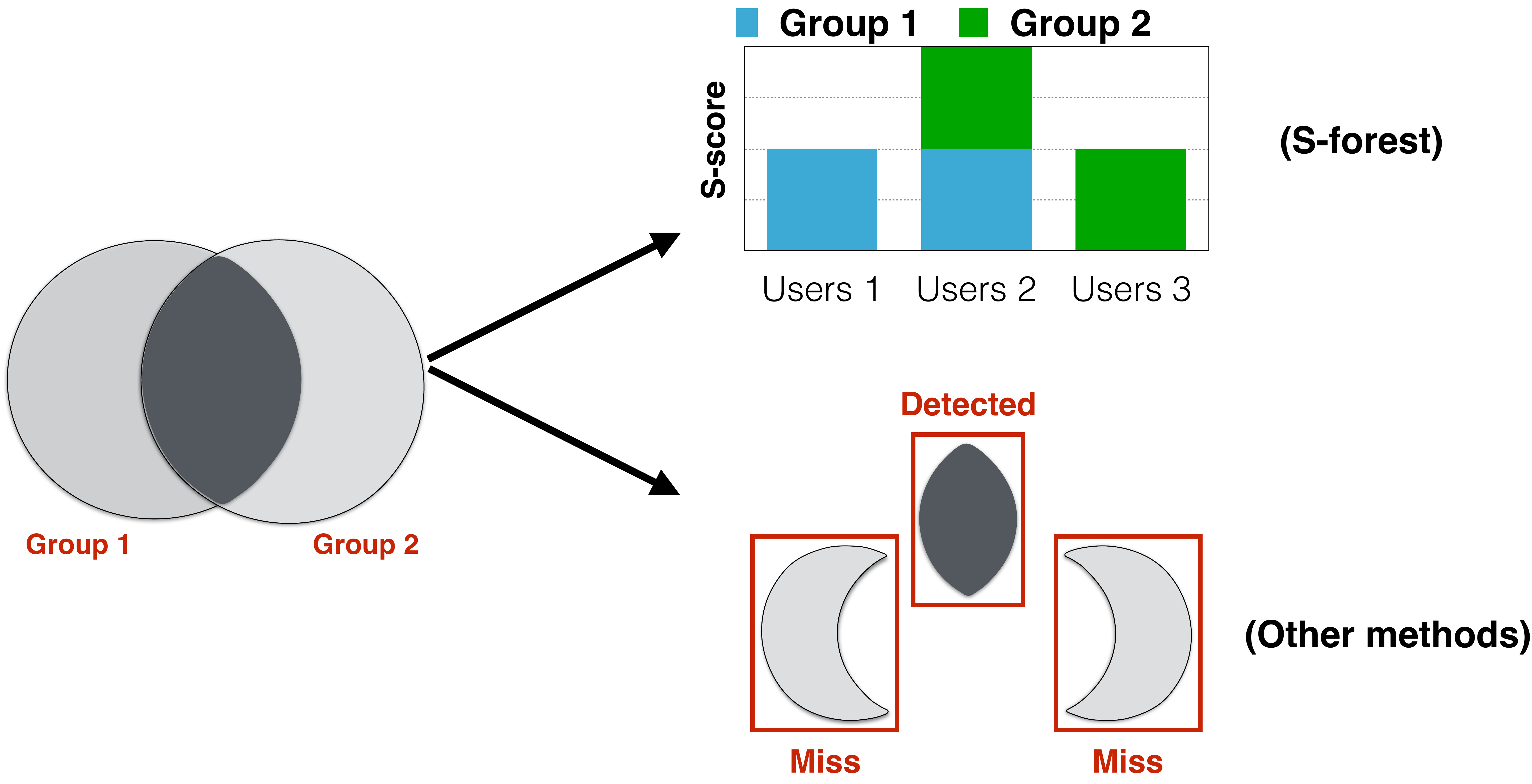}
\centering
\caption{Group overlap} \label{fig:group}
\vspace{-1em}
\end{figure}

\para{Property 2: Automatic Discovery of Critical Features.} 
One of the challenges in group fraud detection is the issue of feature selection. Fraud groups do not show synchrony on \emph{all} features, but only a few key dimensions. However, these key features are not known in advance, and they change with time and fraud groups. A robust detection algorithm must be able to discover such key features automatically. None of the existing approaches such as \cite{MZOOM, CROSSSPOT, DCUBE} performs satisfactorily in this regard. They either rely on hand-crafted feature weighting or treat all features equally, resulting in the inability to adapt to changes or handle exploded feature combinations when the number of dimensions is high.

S-forest, however, greatly mitigates this problem. With S-forest, each feature is examined individually. The final score is aggregated via weighted addition. Information from different features is automatically accumulated without the need for explicit feature selection or combination. As will be shown by experiments in section \ref{exp:sforest}, S-forest performs much more stably on high-dimensional data.

\vspace{-0.5em}
\section{Experiments} 

\vspace{-0.5em}
\begin{table}[h]
	\caption{Datasets used in experiments. `dims' represents `dimensions' }\label{tab:details}
	\vspace{-0.5em}
	\centering
	\begin{tabular}{ccc|ccc}
		\toprule
		      \multicolumn{3}{c}{ S-tree}& \multicolumn{3}{c}{ S-forest}\\ 
		\midrule 
		 datasets & edges & dims & datasets & entries & dims\\
		 \midrule 
		 AmazonOffice\cite{amazon_data} &53K &2 & YelpChi\cite{YELP} &67K & 4    \\
		 AmazonBaby\cite{amazon_data}&160K & 2 & YelpNYC\cite{YELP}& 359K & 4  \\
		 AmazonTools\cite{amazon_data}& 134K & 2 & YelpZip\cite{YELP}& 1.14M & 4  \\
		 AmazonFood\cite{amazon_data} &1.3 M  & 2 & DARPA\cite{DARPA} & 4.55M & 4   \\
		 AmazonVideo\cite{amazon_data} & 583K  & 2& AirFore\cite{AIRFORCE}& 4.89 M  & 8  \\
		 AmazonPet\cite{amazon_data} & 157K & 2 &  Registration  & 26k & 6   \\
		 Twitter\cite{twitter_size} &1.47B &2 & & &    \\
		  \bottomrule
	\end{tabular}
	\vspace{-1em}
\end{table}

We implemented S-tree and S-forest in Python and conducted extensive experiments using both public and proprietary datasets against several strong competing algorithms representing state of the art technology in fraud detection. We use their official open-source implementations for all the competing algorithms. All experiments were run on a server with 2 2.2 GHz Intel Xeon E5 CPUs and 32 GB memory.

Since S-tree and S-forest apply to different datasets, the experiments can be classified into two categories: evaluations of S-tree on 2-dimensional data and S-forest on multi-dimensional data. Table \ref{tab:details} shows the details about the datasets used for each case. We also have to choose different competing algorithms for comparison since they too apply to data with different dimensionalities.
 
\para{Evaluation Metrics.}
We run S-tree and S-forest on a total of 13 datasets. Among them, the Registration is a proprietary dataset from an e-commerce provider's log. The rest are publicly available. The six Amazon datasets and Twitter dataset are unlabeled.    
There are two methods to obtain labels for evaluation. For the Amazon datasets, we perform the standard dense subgraph injection to synthesize fraud groups, which is a commonly-used method in fraud detection research \cite{FRAUDAR,CATCHSYNC}. For the Twitter dataset, we take the same approach as Fraudar \cite{FRAUDAR} and cross-reference known sources to find suspicious accounts. Details will be provided next. Once labels are obtained, we evaluate the performance using standard metric F1-score or AUC.
\vspace{-0.5em}
\subsection{Experiments with S-tree}
\vspace{-0.5em}

We compared S-tree against two algorithms that are applicable to this situation (mining bipartite graph for dense subgraphs) and have achieved the strongest performance so far: Fraudar\cite{FRAUDAR} and CatchSync\cite{CATCHSYNC}. Fraudar adapts the well-known algorithm \cite{Greedy} for the \emph{densest subgraph problem} to the weighted bipartite graph situation, which has been shown to be effective in fraud group detection. However, Fraudar, tries to find a subgraph maximizing the density metric $F = \frac{|\mathcal{M}| \cdot |\mathcal{N}|}{|\mathcal{M}| + |\mathcal{N}|}$. When the sizes of $\mathcal{M}$ and $\mathcal{N}$ are unbalanced (e.g., when $|\mathcal{N}| \gg |\mathcal{M}|$, as is common in many fraud groups), Fraudar performs poorly, since in this case $F \rightarrow  |\mathcal{M}|$ and it is possible that subgraphs induced by legit nodes have bigger values in terms of $F$.  
  
CatchSync detects dense subgraphs leveraging node degree and HITS\cite{HITS}. However, it is vulnerable to camouflaging: both A-Cam and P-Can can change the connectivity between fraud nodes and legit nodes, altering their degrees and HITS's dramatically. 


\para{[Amazon Datasets]\cite{amazon_data}.} We first run the algorithms on six collections of reviews for different types of commodities on Amazon (Table \ref{tab:details}). We designed two subgraph injection schemes: the first is to 
examine in detail the performances for detecting HIM bicliques ($\rho = 1$) and dense subgraphs ($\rho < 1$); the second is for more general performance evaluation.

\para{[Injection Scheme 1].} To simulate the attack models of fraudsters, we use the same approach as \cite{FRAUDAR,CATCHSYNC} and generate datasets by injecting a fraud group with varying configurations into AmazonOffice. The injected fraud group is set as $\mathcal{G}(|\mathcal{N}| = 200, |\mathcal{M}| = \lambda, \rho)$ where $\lambda$ denotes the number of  fraud target nodes and $\rho$ follows Definition \ref{def.rho}. We also introduce two perturbations: (1) unbalanced dense subgraph created by varying  $\lambda$;  and (2) randomly generated A-Cam or P-Cam.  Let $\theta$ denote the number of camouflage edges for each node in $\mathcal{N}$ or $\mathcal{M}$. We set $\theta =\lambda$ or $ \theta =0.5\lambda$ in the experiments.
We set $\rho =1.0$ and $\rho = 0.6$ to represent the different cases of detecting HIM bicliques and dense subgraphs, respectively. 

\subsubsection{Detection of HIM Bicliques ($\rho$ = 1.0) } Fig.\ref{fig:synt1} shows the performance of each algorithm on detecting HIM bicliques. The x-axises are $\lambda$ (varying from 0 to 25), and y-axises are the best-F1 scores. A few conclusions are clear: (1) Without camouflage, both S-tree and CatchSync attain perfect accuracy, even when $\lambda$ is reduced to 2. This is not surprising. Since isolated-unbalanced bicliques must form maximal-shared paths, they can be detected accurately by S-tree. At the same time, since HIM Bicliques'  HITS and degree values are very prominent, CatchSync also obtains the same performance. On the other hand, Fraudar's inability for handling unbalanced dense subgraph is clearly demonstrated by the experiments. (2) Both S-tree and Fraudar show resistance in face of camouflage. In contrast, camouflage destroys CatchSync's effectiveness, as HITs score and degree of fraud source node change dramatically as the number of camouflage edges increases. 

 
\begin{figure}[h]
\vspace{-1em}
\includegraphics[width = 0.8\columnwidth, height = 0.5\columnwidth]{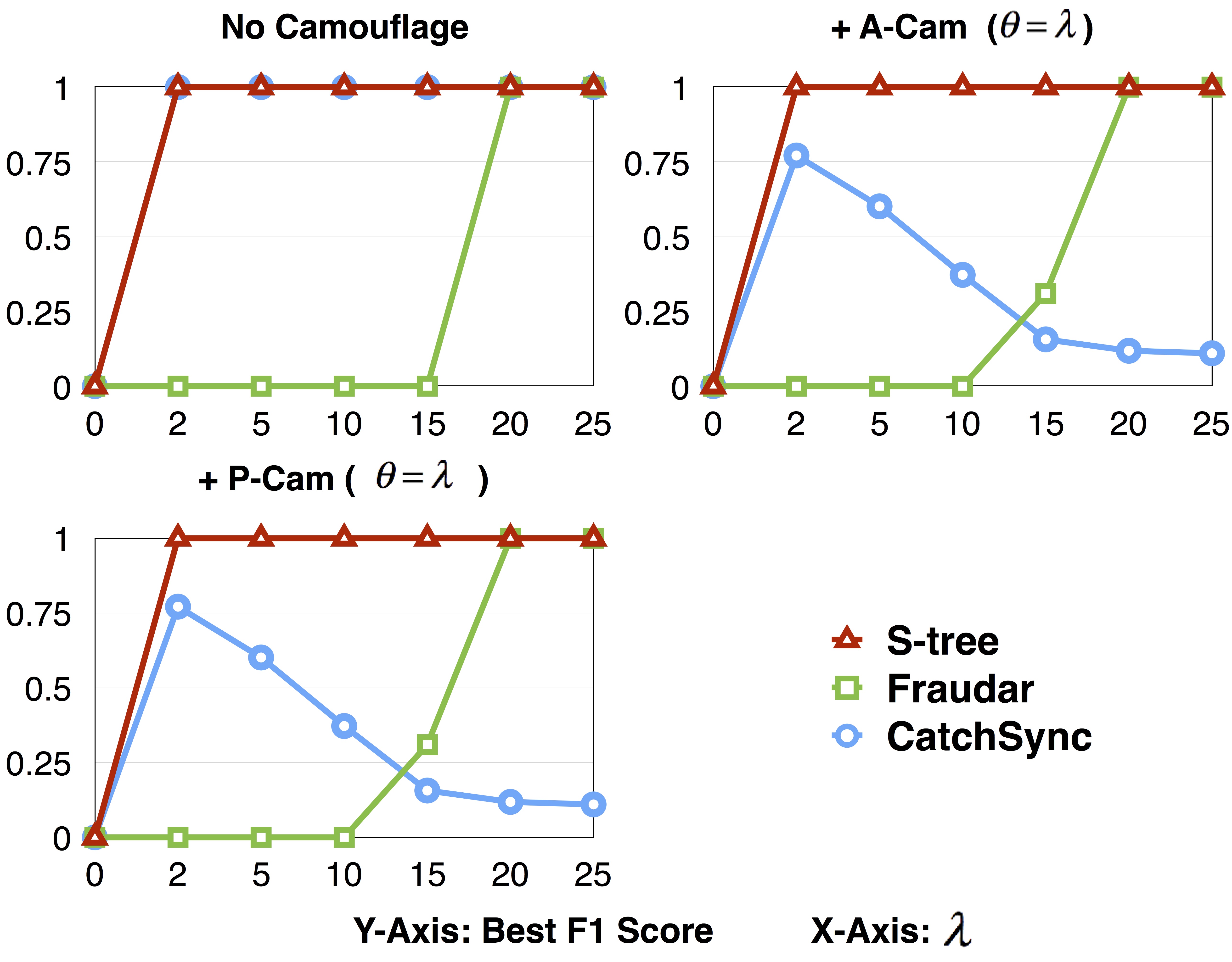}
\centering
\caption{Performance on AmazonOffice ($\rho = 1.0$)} \label{fig:synt1}
\vspace{-1em}
\end{figure}
  
\subsubsection{Detection of Dense Subgraphs: $\rho$ = 0.6} 
When $\rho$ is reduced to 0.6, the performance of all algorithms decrease, as shown in Fig.\ref{fig:synt2}. Not surprisingly, S-tree still outperforms Catchsync and Fraudar in face of camouflage, especially for dense subgraphs with large sizes. Let us take a closer look at a concrete data point: camouflage $\theta = \lambda = 20$. In this case, the average degree of a fraud group is only 11.32, which is lower than a group formed by normal users. This caused a large number of false positives for Fraudar. The number of camouflage edges of each user reaches 20 so that its degree and HITs score have no difference with normal ones, which accounts for the deterioration of CatchSync's performance. For S-tree, however, 20 baskets of fraud target nodes are mapped into two shared subtrees in S-tree, resulting in anomalously high suspiciousness scores and almost perfect accuracy. 

\vspace{-1em}
\begin{figure}[h]
\includegraphics[width = 0.8\columnwidth, height = 0.7\columnwidth]{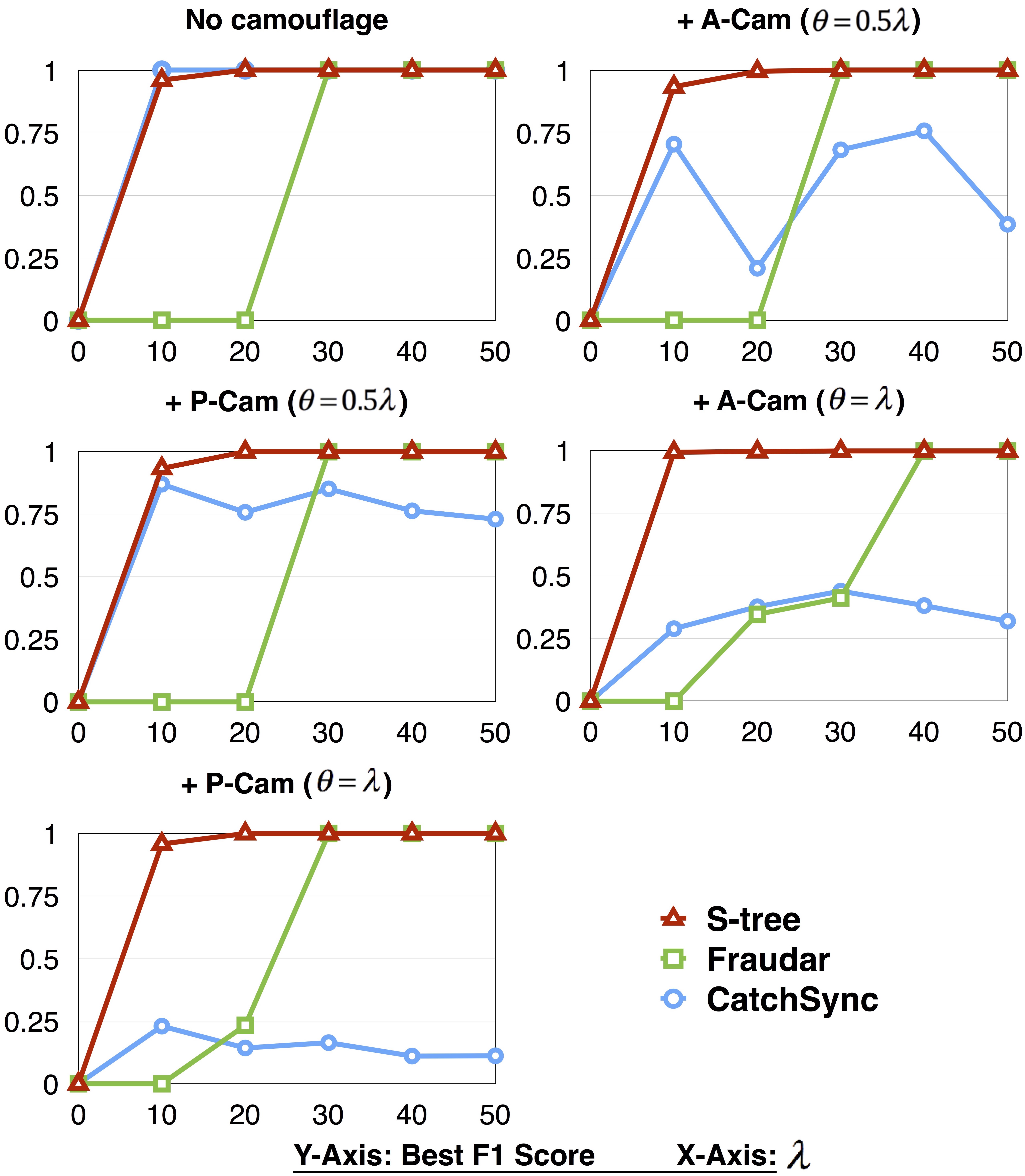}
\centering
\caption{Performance on AmazonOffice ($\rho = 0.6$)} \label{fig:synt2}
\vspace{-1em}
\end{figure}

\para{[Injection Scheme 2].} In this experiment, we inject 10 fraud groups $\mathcal{G}( \mathcal{N}| = 200, \mathcal{M} = \lambda, \rho)$ into AmazonOffice, AmazonBaby,  AmazonTools, AmazonFood, AmazonVideo, and AmazonBook. $\lambda$ and $\rho$ are randomly chosen from $[5, 50]$ and $[0.6, 1.0]$, respectively. Out of the 10 fraud groups, 3 of them are augmented with A-Cam, and another 3 augmented with P-Cam. In all the camouflage cases we set $\theta = \lambda$. The results are shown in Table \ref{tab:amazon}. Overall, S-tree is the most robust and performs best across all variations and camouflages. 
\vspace{-0.5em}
\begin{table}[h]
	\caption{Performance(AUC) on the [Amazon] datasets }\label{tab:amazon}
	\centering
	\begin{tabular}{c|cccc}
		\toprule
		& AmazonOffice & AmazonBaby& AmazonTools \\
		\midrule  
		Fraudar &0.9415 &0.8674 & 0.9264  \\
		 CatchSync & 0.8412 &0.8190& 0.7907  \\
		\midrule  
		S-tree & \textbf{0.9987} & \textbf{0.9595} & \textbf{0.9389}\\
		\toprule
		&AmazonFood &AmazonVideo & AmazonBook \\
		\midrule  
		Fraudar & 0.6915 & 0.7361 & 0.8923  \\
		 CatchSync & 0.7412 &0.6990& 0.7634  \\
		\midrule  
		S-tree & \textbf{0.8558} & \textbf{0.8756} & \textbf{0.9435}\\
		\bottomrule
	\end{tabular}
	\vspace{-1em}
\end{table}

\para{[Twitter Dataset\cite{twitter_size}].} To demonstrate the effectiveness of S-tree on a very large social graph, we run it on the Twitter Dataset \cite{twitter_size}, a graph containing 41.7 million users and 1.47 billion follows over July 2009. The dataset does not contain any labels.  

First of all, the algorithm is very fast. It took only ten hours to finish on a single server with two CPUs.

%

S-tree detects a suspicious group with 5531 followers. To further validate the group, we randomly sample 150 followers for hand labeling to determine how many of them appear fraudulent.
To do this, we label which users are fraudulent based on the following characteristics:
first, we check how many accounts in this group are deleted or suspended (penalized by Twitter)~\footnote{https://tweeterid.com/}; second, we inspect Tweets of these accounts that whether contain the URLs of two known follower-buying services, TweepMw, and TweeterGetter, which shows if they had involved in fraudulent activities. We sample 200 followers and investigate of their profiles, 43.2\% of detected accounts(followers) have tweets advertising TweepMe or Tweetergetter, and additional 34.6\% are deleted, protected or suspended. In comparison, we choose two non-sharing subtrees of S-tree using the same investigation, we find that none of the accounts has tweets advertising and only a few accounts are deleted. Thus, the results further support the effectiveness of S-tree.



\vspace{-0.5em}
\subsection{Experiments with S-forest} \label{exp:sforest}
\vspace{-0.5em}

For multi-dimensional data, we compare S-forest against four strong rivals. CrossSpot\cite{CROSSSPOT}, M-Zoom\cite{MZOOM}, M-Biz\cite{MBIZ}  and D-Cube \cite{DCUBE} all detect dense blocks in tensors. For CrossSpot, we set the random initialization seeds as 400, 500 and 600. For M-Zoom and D-Cube, we use three metrics they supported in \cite{MZOOM,DCUBE} and choose the one with the best accuracy. We evaluate each methods' performance in terms of ranking of suspiciousness of source nodes. Experiments were conducted on a number of datasets described below.

\para{[DARPA]} DARPA\cite{DARPA} was collected by the Cyber Systems and Technology Group in 1998. It is a collection of network connections, some of which belong to TCP attacks. Each connection contains \emph{source IP, target IP, timestamp.} Thus the dataset can be model as $D$(connection, source IP, target IP, timestamp). The dataset includes labels indicating whether each connection is malicious or not. We run CrossSpot, M-zoom, and D-cube on the tensor formed by $D$. Table \ref{tab:darpa} gives the performance of each method for detecting malicious connections. S-forest evidently surpasses the best previous accuracy created by D-Cube.

\para{[AirForce]} AirForce \cite{AIRFORCE} is a dataset used for KDD Cup 1999, which has been used in M-Zoom and D-Cube.  It includes a wide variety of intrusions simulated in a military network environment. However, it does not contain any specific IP address.  Following D-Cube\cite{DCUBE}, we also use the same seven features and model it as an $1+K$ dataset: $D$(connection, protocol, service, src bytes, dst bytes, flag, host count, srv count) \footnote{Please see http://www.cs.cmu.edu/\~{}kijungs/codes/dcube/supple.pdf for detailed description of the fields.}. The dataset includes labels indicating whether each connection is malicious or not. Table \ref{tab:darpa} presents the result of each method. S-forest achieves the best accuracy.
\vspace{-0.5em}
\begin{table}[h]
	\caption{Performance(AUC) on the [DARPA] and [AirForce] dataset}\label{tab:darpa}
	\centering
	\begin{tabular}{c|cc}
		\toprule
		Dataset & DARPA & AirForce\\
		\midrule  
		 CrossSpot &0.923 & 0.924 \\
		 M-Zoom & 0.923 & 0.975 \\
		 M-Biz &0.923 & 0.975 \\
		 D-Cube & 0.930 & 0.987 \\
		\midrule  
		S-forest & \textbf{0.942} & \textbf{0.991}\\
		\bottomrule
	\end{tabular}
	\vspace{-1em}
\end{table}

\para{[Yelp]}\cite{YELP}. YelpChi, YelpNYC, and YelpZip are three datasets collected by \cite{YELP2} and \cite{YELP} containing review actions for restaurants on Yelp. They all can be modeled as $D$(restaurant, user, rating-star, timestamp). In these datasets, our task is to detect fraudulent restaurants that purchase fake reviews.  The three datasets all include labels indicating whether each review is fake or not. A restaurant is labeled as ``fraudulent'' if the number of fake reviews it receives exceeds 40. Table \ref{tab:yelp} shows the results.
\vspace{-0.5em}
\begin{table}[h]
	\caption{Performance(AUC) on the [YELP] datasets}\label{tab:yelp}
	\centering
	\begin{tabular}{l|ccc}
		\toprule
		&YelpChi & YelpNYC & YelpZip\\
		\midrule  
		CrossSpot & 0.9905  & 0.9328 & 0.9422 \\
		M-Zoom &0.9850 & 0.9255 & 0.9391 \\
		M-Biz & 9850 & 0.9229 &  0.9302\\
	    Dcube & 0.9857 & 0.9232  & 0.9391\\
		\midrule  
		S-forest & \textbf{0.9945}& \textbf{0.9406}& \textbf{0.9456} \\
		\bottomrule
	\end{tabular}
	\vspace{-1em}
\end{table}

\para{[Registration Dataset]} 
To test S-forest's performance on real-world fraud detection task, we obtained real registration data from a major e-commerce provider. The data contains 26k registration entries sampled from the vendor's production system during a three-day period. Each entry contains an account id and multiple features such as IP subset, phone, and osType. 
Some of the accounts were registered by bots and were later sold to fraudsters who then conducted various frauds using them.  These sampled accounts were tracked for several months and were labeled fraudulent if they were found to engage in frauds by human experts. A total of 10k fraud accounts were identified.
 
To verify the automatic key feature discovery capability of S-forest, we did very little feature engineering and retained most information from the log. We formulated data as $1+5$ model: $D$(account, phone-prefix, IP-prefix, osType, phone-prefix-3, is-same-province), where is-same-province is a Boolean value indicating if the IP and phone of a user come from a same province or not. Among the five features, phone and IP are two critical features where only malicious accounts exhibit strong resource sharing pattern, while osType, phone-3, and is-same are three noise features where both legit and malicious accounts form dense subgraphs. We conducted five experiments, each with different feature sets. Results are shown in Table \ref{tab:real}. 
  
Clearly, S-forest outperforms other baselines by large margins across all configurations. Specifically, when the noise features are added,  the performance of CrossSpot, M-Zoom, M-Biz, and D-Cube deteriorates while S-forest stays stable. This is due to the fact that these algorithms did not differentiate the importances of different features in identifying frauds and were misled by clustering on some irrelevant features. S-forest, on the other hand, is capable of automatically focusing on critical features while ignoring noisy perturbations.   

The closer investigation uncovers more detailed differences. First, we found that the sizes of fraud groups vary widely, ranging from tens to thousands. Other algorithms miss some groups with small sizes, since they appear to be  low-density to them. One particular fraud group consists of 75 accounts that use the same IP subnet. The group forms an extremely unbalanced dense subgraph since the number of target nodes (IP subnet) is only 1. The other algorithms all fail to detect it.

\vspace{-0.5em}
\begin{table}[ht] 
\caption{Performance (AUC) on [Registration] dataset. `C' represents `critical feature' and `N' represents `noise feature'.}\label{tab:real}
\centering
\begin{tabular}{c|c|c|c|c|c}
\toprule
 &  1 C  & 2 C & 2C +1N  & 2C + 2N & 2C+3N \\
 \midrule
CrossSpot  & 0.7471 & 0.7312 &0.7307  &0.7539  & 0.8099\\
M-zoom  & 0.7544  &0.8880 &0.8409 &0.8754  & 0.6802\\ 
M-Biz &0.7577 & 0.8842 &0.8409 & 0.8674& 0.6812\\
D-Cube &0.7601 &0.9201& 0.8532 &0.8723 & 0.7102\\
 \midrule
S-forest &\textbf{0.7678}  &\textbf{0.9961} &\textbf{0.9964} & \textbf{0.9876} &\textbf{0.9976}\\
\bottomrule
\end{tabular}
\vspace{-0.5em}
\end{table}

Second, another example confirms our earlier analysis that other algorithms have difficulty handing overlapping groups. We found in the data that there is one fraud group $A$ formed on IP subnet and another group $B$ formed on phone-prefix-7. There are some overlapping users who share both two features. S-forest correctly detected both $A$ and $B$ while other methods missed $A \setminus A \cap B$.


\vspace{-0.5em}

\subsection{Scalability and Efficiency}
\vspace{-0.5em}

\vspace{-0.5em}
\begin{figure}[h]
\includegraphics[width = 0.8\columnwidth, height = 0.5\columnwidth]{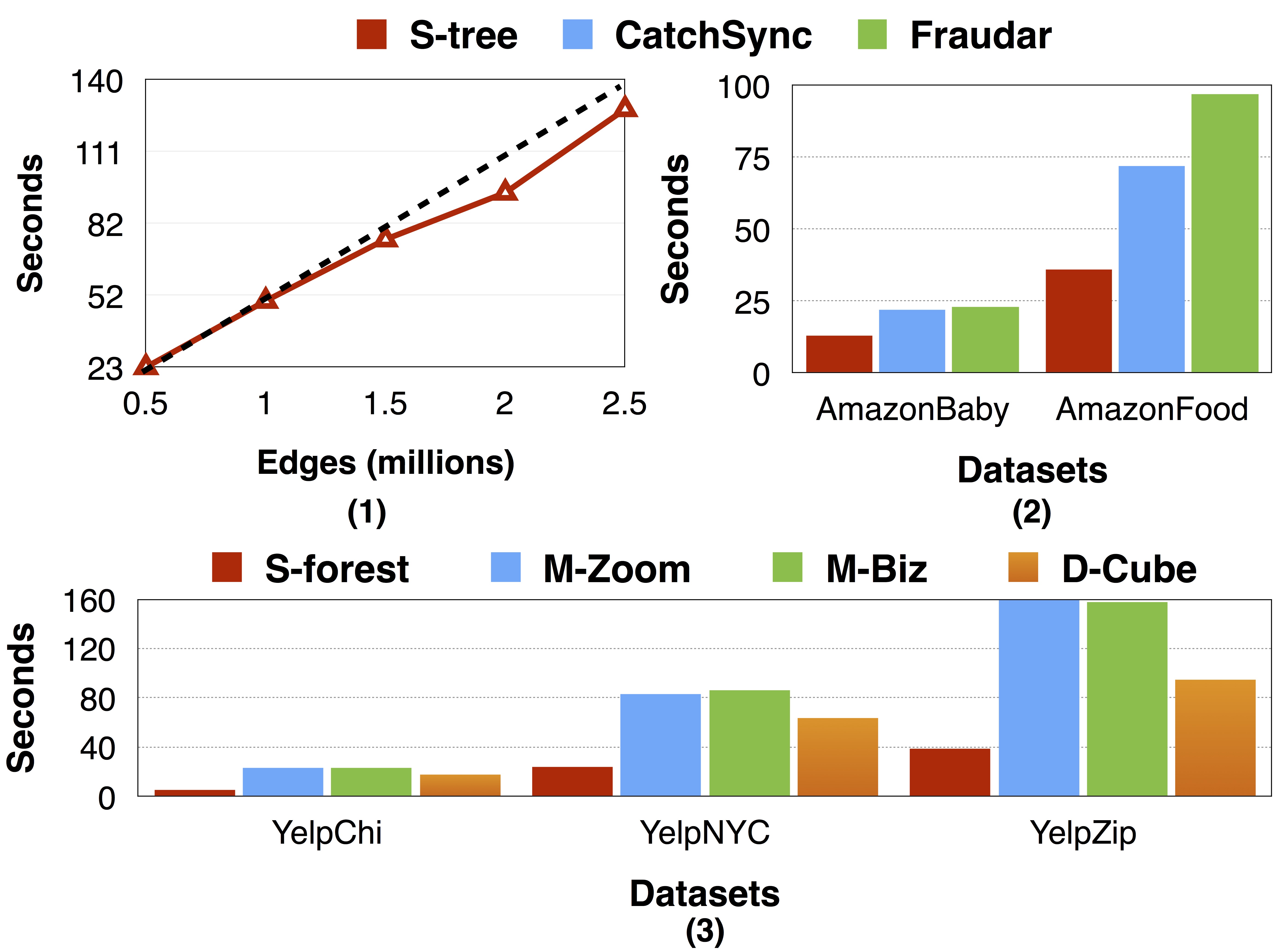}
\centering
\caption{S-tree/S-forest has linear-time complexity and is faster than other state-of-the-art methods} \label{fig:time}
\vspace{-1em}
\end{figure}

As mentioned before, building S-tree takes $O(2|E|)$ time and all algorithms only need one scan of S-tree. To verify this linear complexity, we recorded S-tree's running time on AmazonApp data \cite{amazon_data}. We vary the scale by subsampling the edges. Fig.\ref{fig:time} (1) shows the linear scaling of S-tree's running time in the number of edges. In addition, for comparison, we show the running time of S-tree, CatchSync and Fraudar on AmazonBaby and AmazonFood datasets in Fig.\ref{fig:time}(2). S-tree is faster than the two, especially when the number of edges is large. In fact, Fraudar needs $O(|E|log|N|)$ time for detecting a single fraud group (and there typically exist multiple fraud groups in a dataset) and CatchSync needs to calculate HITS score. Both are slow for large-scale data. (3) S-forest also is linear scaling with entries of datasets, and is much faster than other state-of-the-art methods.

\section{Related Work}

\para{Algorithms for maximal biclique.} Enumerating all maximal bicliques cannot be solved in polynomial time \cite{eppstein1994arboricity} and people often place restrictions on it to find maximal bicliques efficiently.  Algorithms for finding maximal bicliques in a bipartite graph can be classified into three lines. The first is exhaustively search-based approaches \cite{sanderson2003obtaining, mushlin2007graph}: build all subsets of one vertex partition, find their intersections in the other partition, and check each for maximality. However, they must generally place restrictions on the problem to stem the enormous search space. The second line relies on graph inflation \cite{makino2004new}: transform the problem into finding all maximal cliques in a graph by adding possible edges. The third line is based on association mining.  \cite{li2007maximal} and \cite{zhang2014finding} find all maximal bicliques respectively based on \cite{minging5} and \cite{bron1973algorithm}. Certainly, their computational cost is huge. Moreover, biclique can be transformed into frequent itemsets in transactional databases. Thus, frequent itemset mining methods \cite{minging1, minging2,minging3,minging5} may be helpful in finding maximal bicliques. However, they cannot solve maximal bicliques or MHI bicliques enumeration problem in polynomial time. 
Various restrictions on either inputs or outputs are proposed to find maximal bicliques, including bounding the maximum input degree\cite{tanay2002discovering}, bounding an input’s arboricity\cite{eppstein1994arboricity}, bounding the minimum biclique size\cite{sanderson2003obtaining}, and figure-of-merit\cite{mushlin2007graph}. Naturally, the algorithms relying on these restrictions do not account for MHIBP.

\para{Mining dense subgraphs for fraud detection.}
Mining dense subgraphs \cite{FRAUDAR, KCORE, SPOKEN, NETPROBE, FRAUDEAGLE} is effective for detecting the fraud groups of users and objects connected by a massive number of edges. Fraudar \cite{FRAUDAR} tries to find a subgraph with the maximal average degree using a greedy algorithm. ~\cite{INFERRING,FBOX} adapt singular-value decomposition (SVD) to capture abnormal dense user blocks. Furthermore, Belief propagation \cite{FRAUDEAGLE, NETPROBE}, HITS-like ideas \cite{COMBATING, CATCHSYNC,UNDERSTAN} are all adapted to detect high-density signals of fraud groups in a graph. FraudEagle \cite{FRAUDEAGLE} uses the Loopy Belief Propagation to assign labels to nodes in the network represented by Markov Random Field. 

\para{Detecting high-density signal in tensors.} Multimodal data can be treated as tensors and many works mine the tensors directly for dense blocks. \cite{SYNCHROTRAP, COPYCATCH} discover time constraint of fraudsters and provide a way to combine multiple features to spot fraud. \cite{MAF} spots dense blocks using CP decomposition\cite{CP}. However, as observed in \cite{CROSSSPOT}, using tensor decomposition techniques usually find blocks with low density, and are outperformed by search-based methods.
CrossSpot \cite{CROSSSPOT} starts the search from random seeds and greedily add values into the block until it reaches the local optimum. M-Biz \cite{MBIZ} utilizes similar methods but adjust the block by adding or deleting values from it. M-Zoom\cite{MZOOM} and D-Cube\cite{DCUBE} are both find the densest block by greedily deleting values from the tensor until it reaches the maximal value in terms of a density metric.

\section{Conclusion}

In this paper, we introduced and addressed a new, restricted biclique problem (MHIBP), motivated by real malicious pattern in fraud campaign. In addition to the MHIBP solver, we proposed a practical algorithm that focuses on near bicliques, which is applicable to catching fraud groups in a wide range of situations. Our algorithms are based on two novel structures S-tree, and its extension, S-forest. The problem and algorithms may have interesting applications in other areas such as bioinformatics and social network analysis which we plan to pursue as future work.

\bibliographystyle{IEEEtran}
\bibliography{reference}
\end{document}